\documentclass[letterpaper, 10 pt, conference]{ieeeconf}

\usepackage{amsmath}

\let\labelindent\relax

\usepackage{amsmath}

\let\labelindent\relax
\usepackage{textcomp}
\usepackage{amsthm}
\usepackage{amssymb}  
\usepackage{mathtools}
\usepackage[utf8]{inputenc}
\usepackage{comment}
\usepackage{algorithm,algpseudocode}
\algblock{Input}{EndInput}
\algnotext{EndInput}
\algblock{Output}{EndOutput}
\algnotext{EndOutput}
\newcommand{\Desc}[2]{\State \makebox[2em][l]{#1}#2}

\usepackage{breqn}
\usepackage{graphicx}
\usepackage[version=4]{mhchem}
\usepackage{siunitx}
\usepackage{longtable,tabularx}
\usepackage{graphics} 
\usepackage{float}
\usepackage{epsfig} 
\usepackage{times} 
\usepackage{color}
\usepackage[dvipsnames]{xcolor}
\usepackage{pifont}
\usepackage{enumitem}
\usepackage{bm}
\usepackage{bbm}
\usepackage{soul}
\makeatletter
\newcommand{\multiline}[1]{%
  \begin{tabularx}{\dimexpr\linewidth-\ALG@thistlm}[t]{@{}X@{}}
    #1
  \end{tabularx}
}
\makeatother
\usepackage{hyperref}
\hypersetup{
    colorlinks=true,
   linkcolor=black,
    filecolor=blue,      
   urlcolor=blue,
    pdftitle={Overleaf Example},
    pdfpagemode=FullScreen,
    }

\DeclareMathOperator{\sinc}{sinc}

\newcommand{\E}{\mathbb{E\,}}

\newcommand{\R}{\mathbb{R}}

\newcommand{\atantwo}{\text{atan2}}

\newtheorem{lemma}{Lemma}
\newtheorem{remark}{Remark}
\newtheorem{theorem}{Theorem}

\DeclareMathOperator*{\argmin}{arg\,min}

\newcommand{\python}{\mbox{\textsc{Python}}}

\title{\LARGE \bf
A Control Approach for Nonlinear Stochastic State Uncertain Systems with Probabilistic Safety Guarantees}

\author{Mohammad S. Ramadan,\,Mohammad Alsuwaidan,\,Ahmed Atallah,\, and Sylvia Herbert
\thanks{The authors are with the Department of Mechanical and Aerospace Engineering, University of California San Diego, La Jolla, CA, USA. E-mails: {\tt\small \{msramada, mnalsuwa, aatallah, sherbert\}@ucsd.edu}}}

\begin{document}

\maketitle
\thispagestyle{empty}
\pagestyle{empty}

\begin{abstract}
This paper presents an algorithm to apply nonlinear control design approaches in the case of stochastic systems with partial state observation. Deterministic nonlinear control approaches are formulated under the assumption of full state access and, often, relative degree one. We propose a control design approach that first generates a control policy for nonlinear deterministic models with full state observation. The resulting control policy is then used to build an importance-like probability distribution over the space of control sequences which are to be evaluated for the true stochastic and state-uncertain dynamics. This distribution serves in the sampling step within a random search control optimization procedure, to focus the exploration effort on certain regions of the control space. The sampled control sequences are assigned costs determined by a prescribed finite-horizon performance and safety measure, which is based on the stochastic dynamics. This sampling algorithm is parallelizable and shown to have computational complexity indifferent to the state dimension, and to be able to guarantee safety over the prescribed prediction horizon. A numerical simulation is provided to test the applicability and effectiveness of the presented approach and compare it to a certainty equivalence controller. 
\end{abstract}
\section{Introduction}
The challenge with output feedback control in stochastic systems is the interaction between control and estimation. Stochastic control deals with this interaction by introducing the concept of the information state \cite{kumar2015stochastic,striebel1965sufficient}, which addresses the propagation of uncertainty together with the state estimate, instead of regarding this estimate as the true state value as in certainty equivalence control. The introduction of the information state results in a significant increase in the complexity and dimensionality of the control problem and leads to algorithms which are prohibitive computationally. On the other hand, to maintain computational tractability, nonlinear control design approaches are typically formulated for deterministic and full-state feedback systems \cite{khalil2002nonlinear}. This difference in the problem formulation makes it difficult to use nonlinear design approaches in stochastic settings. 

In this paper, we propose a random search algorithm to solve the stochastic nonlinear optimal control problem. For computational efficiency, we first employ a nonlinear controller designed for a surrogate deterministic dynamic model. The nonlinear controller and the surrogate dynamics are used as a constructor of an importance-like distribution which reduces the searching space and guides the effort of sampling finite-horizon control sequences. These sequences are then evaluated with respect to a prescribed finite-horizon performance and safety measure when applied to the true stochastic dynamics. The first control action, of the sequence with the highest performance and safety measure is applied at the current time in a receding-horizon fashion.

Out of the many nonlinear design approaches to test the proposed algorithm, this paper uses control barrier functions (CBF) to generate the deterministic control policy. CBFs are considered the ``dual'' of control Lyapunov functions, but they enforce safety rather than stability. Existing literature about CBFs has limited answers to system uncertainties or output feedback, two characteristics that often exist in real-world systems. A method addressing stochastic systems within the CBF framework have been developed in \cite{cosner2023robust}, but for the full state-feedback case. For the case of partial state observation, or output feedback, the CBF-based approaches are typically limited in assumptions to linear systems \cite{wang2021chance}. Furthermore, time-discretization, and possibly state and input space  discretization, are employed while relying on guarantees achieved in continuous settings. This weakens the validity of the inherited guarantees and leads to the conservativeness in the adjusting safe sets by additional safety margins \cite{breeden2021control}.


The key idea of this paper is to use the policy generated by the deterministic CBF controller to guide a random search optimization for the stochastic nonlinear control problem. This bears resemblance to importance sampling, an important technique in the field of Monte Carlo integration \cite{liu2001monte}. This technique approximates an integral by sampling the importance distribution: a chosen probability distribution that is ideally concentrated in regions of high contribution to the integral in-hand. The analogy between our proposed algorithm, and importance distribution, is in constructing a distribution in the control space that is concentrated in regions likely to be of low cost and guaranteed safety. This is similar to sampling methods which rely on an information theoretic construction of the importance distribution \cite{williams2018robust} for full state-feedback systems. However, our approach does not require full state-feedback. Moreover, in contrast to sampling approaches based on Monte Carlo integration, our approach does not rely on the laws of large numbers; a finite number of samples/scenarios is enough to guarantee safety with high probability.

\section{Problem Formulation} \label{section: ProblemFormulation} 
Consider the following discrete-time state-space model
\begin{subequations}
\begin{align}
x_{k+1}&=f(x_k,u_k,w_k),\label{eq:NonSys_a}\\
y_k&=g(x_k,v_k),\label{eq:NonSys_b}
\end{align}\label{eq:NonSys}
\end{subequations}
where the state is $x_k\in\mathbb R^{r_x}$, control input $u_k\in\mathbb R^{r_u}$, and exogenous disturbances $w_k$ and $v_k$. The initial state $x_0$ has distribution $p_0$ represented by a set of $L$ equiprobable particles $\Xi=\{x_{0\mid 0, j} \mid j=1,\hdots,L\}$. The stochastic disturbance processes, $\{w_k\}$ and $\{v_k\}$, are white, possess known densities $\mathcal W$ and $\mathcal{V}$, and are independent from each other and each from $x_0$.

We assume we are given the following cost function
\begin{align*}
    J(p_0, \{ u_k\}_{k=0}^{N-1}) =\E \left (\gamma^N\ell_N(x_N) + \sum_{k=0}^{N-1} \gamma^k \ell_k(x_k,u_k) \right),
\end{align*}
then, the objective is to find the control sequence $\{ u_k\}_k^{N-1}$ that solves the following optimization problem:
\begin{equation}
\begin{aligned} \label{eq:optimizationProb}
    \min_{\{ u_k\}_{k=0}^{N-1}} &J(p_0, \{ u_k\}_{k=0}^{N-1}),\\
    &\text{subject to, for all }k,\\
    & \mathbb P (x_k \in \mathcal C) \geq 1-\epsilon,\, u_k \in \mathbb U,
\end{aligned}
\end{equation}
where: $\gamma \in (0,1]$ is the discount factor; $N$ is the prediction horizon; the expectation $\E$ is with respect to a probability space of: elements $(x_0,\{w_k\}_{k=0}^{N-1})$, the suitable product Borel $\sigma-$field, and the probability measure $\mathbb P$; $\ell_k$s are the stage cost functions and $\E \lvert \ell_k(x_k,u_k) \rvert<\infty$\footnote{Possible assumptions to satisfy this: $\ell_k$s are bounded; or for each $u_k$, $\ell_k(\cdot,u_k)$ is positive and bounded by a quadratic function in $x_k$, the covariances of $w_k$, $x_0$ are finite and $f$ is uniformly Lipschitz.}, for all $u_k$; $\epsilon \in [0,1)$ is the acceptable constraint violation probability and is typically small $\epsilon \approx 0$; the set $\mathcal C$ can be seen as the safe set and is not to be violated in a probability less than $1-\epsilon$.

In the case of partial state observation as in \eqref{eq:NonSys_b}, where in general $y_k \neq x_k$, the concept of the information state appears \cite{kumar2015stochastic}, which is the conditional state density. Designing a control law over the function of all possible densities was concluded very early in the control literature to be prohibitive \cite{striebel1965sufficient}. One of the common approaches is to ignore the state uncertainty and use an estimate of the state as the true state in a feedback control law. This approach is valid only in simple examples like linear-quadratic-Gaussian (LQG) control, where the separation principle holds \cite{aastrom2012introduction}. In this paper, our approach to solve \eqref{eq:optimizationProb} takes into account the state uncertainty and the disturbances, and provides probabilistic guarantees for the original system \eqref{eq:NonSys}.

\section{Prerequisites}
\label{section: Preliminaries} 
Our methodology consists of two major components: A CBF controller (for the surrogate deterministic dynamics) and a particle filter to track the state conditional density. This section gives a brief introduction of each, before we introduce the algorithm in the next section.

\subsection{Control Barrier Functions}
Although the dynamics used in the CBF and safety literature are mostly in continuous-time, the discrete-time dynamics can be obtained by some integrator or simply a zero-order-hold. Thus, we first assume deterministic dynamics and full state accessibility to derive a CBF controller through a surrogate continuous-time formulation. Then we apply a zero-order-hold, accounting for partial observations and exogenous disturbances, to get back the formulation \eqref{eq:NonSys}.

The following continuous-time nonlinear dynamic system conforms to assumptions concerning the existence and uniqueness of its solutions,
\begin{align}
    \dot x = F(x,u), \label{CBF:dynamics}
\end{align}
where $x \in \R ^{r_x}$ is the state and $u \in \mathbb U \subset \R^{r_u}$ is the control input. This system is to be made safe, in the sense of remaining in a control invariant set $\mathcal{C} \subset \R^{r_x}$ that is a subset of the state constraints. 

Control barrier functions are real-valued functions over the state space that encode both the state constraints and long-term effects of the dynamic system.  A CBF has two key properties relevant to safety: its value at a given state provides a measure of safety, and its gradient informs the set of control actions that will preserve safety \cite{ames2019control,chen2020guaranteed}. 
Let $h(x)$ be a smooth real-valued function such that  $\mathcal{C}$ is the superlevel set of $h$, that is, $\mathcal{C} = \{ x \in \R^{r_x} \mid h(x) \geq 0\}$, and let $\alpha: \R \to \R$ be a class $\mathcal K$ function \cite[p.~144]{khalil2002nonlinear}. Then, the function $h$ is called a CBF if 
\begin{align}
    \exists u \in \mathbb U \text{ such that } \frac{\partial h}{\partial x} \cdot F(x,u) = \dot h \geq -\alpha(h),\, \forall x. \label{eq:CBF_condition}
\end{align}
The above condition enforces that the system remains within the control invariant set $\mathcal C$, and therefore preserves non-negative safety values.

The CBF is typically paired with a background performance controller that is designed according to a different safety-blind objective \cite{buch2021robust,chen2020guaranteed}. Online, the CBF acts as a safety filter to minimally modifies the performance control in order to maintain safety:
\begin{equation}
\begin{aligned}
    &u^\star(x) = \argmin_{u \in \mathbb U}\,\lVert u-u_0 \rVert^2, \\
    &u\text{ satisfies the inequality in \eqref{eq:CBF_condition}}. \label{eq:uStar_equation}
\end{aligned}
\end{equation}
The variable $u_0$ represents the control value chosen by this background controller, and $u^\star$ is the modified safe control. If the dynamic model \eqref{CBF:dynamics} is affine in $u$, \eqref{eq:uStar_equation} is a quadratic program and can be efficiently solved online.

\subsection{Particle filtering}
For general nonlinear dynamics of the form \eqref{eq:NonSys}, the particle filter is a convenient approximation to the Bayesian filter \cite{doucet2000sequential}. We continue assuming, as in \eqref{eq:NonSys}, that the current time ($k=0$) state filtered density $p_0$ is provided by a set of equiprobable particles $\Xi=\{\xi^j_0\in\mathbb R^{r_x}, j=1,\dots,L\}$. The future times' particles are then achieved by the particle filter algorithm. 

The dynamics in \eqref{eq:NonSys} have an equivalent representation using transition and measurement densities \cite{schon2011system}: $p(x_{k+1}\mid x_k,u_k)$ and $p(y_k \mid x_k)$\footnote{The notation $p(a\mid b,c)$ denotes the density function of $a$ given the values $b$ and $c$.}, respectively. These densities are characterized by $\mathcal{V}$ and $\mathcal{W}$, and the dynamics $f$ and $g$ in \eqref{eq:NonSys}. 

The initial state density $p_0$ is defined by the particle set $\Xi=\{ x_{0\mid 0, j} \mid j=1,\hdots,L\}$, which is assumed to be given. The two major steps of a basic (bootstrap) particle filter are
\begin{enumerate}
    \item The time update step, in which the set of particles $\{x_{k|k,j}\}_{j=1}^L$ are propagated through the state equation \eqref{eq:NonSys_a}, with a chosen $u_k$, using $L$ realisations of the disturbance $w_{k,j}\sim \mathcal{W}$. The resulting particles are denoted $\{x_{k+1|k,j}\}_{j=1}^L$.
    \item\label{Step2} The measurement update step, wherein if $y_{k+1}$ becomes available, consists of computing the importance weights $\{\Omega_{k+1,j}\}_{j=1}^L$ by:
        $$\Omega_{k+1,j}=\frac{p(y_{k+1}\mid x_{k+1|k,j})}{\sum_{j=1}^{L} p(y_{k+1}\mid x_{k+1|k,j})},\,j=1,\hdots,L.$$
\end{enumerate}
In this paper, we implement a resampling step after Step-(\ref{Step2}), for every time step. The resulting equiprobable particles are denoted $\{x_{k+1 \mid k+1}\}_{j=1}^L$. The conditional mean at any time $k$ can be found by evaluating the sample average $$\E[x_k\mid y_0,\hdots,y_k]\approx \frac{1}{L}\sum_{j=1}^{L}x_{k|k,j}.$$

\section{Control importance-like distribution for random search}
\label{section:main } 

Our random search approach to solve \eqref{eq:optimizationProb} is influenced by importance sampling in Monte Carlo literature, which prioritizes sampling regions in the integration space of higher contribution. We construct an importance-like distribution $\mathcal{U}$ of control sequences $\{u_k\}_{k=0}^{N-1}$, that is ideally concentrated at regions with high confidence to contain feasible, and possibly optimal solutions of \eqref{eq:optimizationProb}.

Sampling the distribution $\mathcal{U}$ is done through a simulation-based algorithm. We first simulate the sequence $\{x_k'\}_k$ characterized by an initial state $x_0 \sim p_0$ and a CBF controller applied in feedback $u^\star_k = \pi(x_k)$ obtained from \eqref{eq:uStar_equation}, in
\begin{align}
    x_{k+1}'=f(x_k',\pi(x_k'),w_k'),\, x_0' \sim p_0,\,w_k'\sim \mathcal{W}, \label{eq:x_k'def}
\end{align}
 where the disturbance $\{w_k'\}_k$ is an i.i.d. realization of $\{w_k\}_k$. 
 
 This sampling procedure defines a distribution of the state sequence $\{x_k'\}_{k=0}^{N-1}$ which in turn defines a distribution of the control sequence $\{u_k^\star\}_{k=0}^{N-1}=\{\pi(x_k')\}_{k=0}^{N-1}$. We call this the importance distribution $\mathcal{U}$. Notice that this distribution is a function of: the original dynamics \eqref{eq:NonSys_a}, the control law $\pi(\cdot)$ (which is based on the chosen surrogate dynamics \eqref{CBF:dynamics}), and the initial state density $p_0$. Our Control Importance Distribution Algorithm (CIDA), shown in Algorithm~\ref{algorithm:CIDA}, outlines\footnote{The function $1(\cdot)$, used in Algorithm~\ref{algorithm:CIDA}, is the indicator function and returns $1$ if $\cdot$ is true and $0$ if false.} this random search procedure we propose to solve \eqref{eq:optimizationProb}.

\begin{algorithm}[h]
\caption{Control Importance Distribution Algorithm (CIDA)}\label{algorithm:CIDA}
\begin{algorithmic}[0]
\Input
\Desc{$p_0$}{\multiline{current time, $k=0$, particle filter's density described by equiprobable particles $\Xi$}}
\Desc{$\pi$}{\multiline{CBF safe controller \eqref{eq:uStar_equation} of the surrogate deterministic dynamics defined in \eqref{CBF:dynamics}}}
\Desc{M}{number of simulations per control sequence}
\Desc{$\alpha$}{acceptable statistical violation rate, $\alpha \in[0,\epsilon) $}
\Desc{R}{\multiline{rollouts number, i.e. the number of control sequence random search trials}}
\EndInput
\Output
\Desc{$u^ \#$}{current time control input}
\EndOutput
\For{$i=1,2,\hdots,R$}
\State  \multiline{Using $p_0$ and $\pi$, sample $\{x_k'\}_{k=0}^{N-1}$ using \eqref{eq:x_k'def}, then record $\{u_k^\star\}_{k=0}^{N-1}=\{\pi(x_k')\}_{k=0}^{N-1}$;}
\State  \multiline{ using $\{u_k^\star\}_{k=0}^{N-1}$, $M-$(iid) simulations of $\{w_k\}_{k=0}^{N-1}$, denoted $(\{w_{k,q}\}_{k=0}^{N-1})_{q=1}^M$, and the state equation \eqref{eq:NonSys_a}, record the $M-$trajectories $(\{x_{k,q}\}_{k=0}^{N})_{q=1}^M$ such that}
\begin{align*}
    x_{k+1,q}&=f(x_{k,q},u_k^\star,w_{k,q});
\end{align*}
\If{\begin{align} \label{eq:statisticalFeasibility}
a_k=\frac{1}{M}\sum_{q=1}^M 1(x_{k,q} \in \mathcal{C}) \geq 1-\alpha, \forall k\end{align} \hskip 4mm}
\State  $\{u_k^\star\}_{k=0}^{N-1}$ is feasible, and calculate $J_i$
\begin{align*}
    J_i =\frac{1}{M}\sum_{q=1}^{M} \left (\gamma^N\ell_N(x_{N,q}) + \sum_{k=0}^{N-1} \gamma^k \ell_k(x_{k,q},u_k^\star) \right );
\end{align*}
\Else 
\State  $\{u_k^\star\}_{k=0}^{N-1}$ is infeasible, and $\,J_i \gets \infty$;
\EndIf
\EndFor
\State If finite, find $\min_{i} J_i$ and its corresponding control sequence $\{u_k^\star\}_{k=0}^{N-1}$, then set $u^\# \gets u_0^\star$;
\State Propagate particle filter to next time and reset $k \gets 0$;
\end{algorithmic}
\end{algorithm}

\subsection{Computational complexity of CIDA}
Notice that in the original problem formulation \eqref{eq:optimizationProb}, the safety condition $\{x_k \in \mathcal{C}\}$ has to be satisfied with probability $\geq 1- \epsilon$. However, Algorithm~\ref{algorithm:CIDA} imposes the safety condition statistically for the $M-$samples, and with rate $\geq 1-\alpha$. Relying only on the laws of large numbers for guarantees demands $M \to \infty$. Similar to \cite{ramadan2023state,luedtke2008sample}, and the randomized sample/scenario framework \cite{campi2019scenario}, we can provide acceptable safety guarantees with relatively small\footnote{In the numerical simulations section we use $M=150$.} number of samples. 
\vskip 3mm
\begin{lemma} \label{Hoeffdings}
    (\textbf{Hoeffding's inequality \cite{hoeffding1994probability}}). For independent random variables $Z_q,\,q=1,\hdots,\bar M$, $\mathbb{P}(Z_q \in [a_q,b_q])=1$, $a_q\leq b_q$, for all $t\geq0$
    \begin{align*}
        \mathbb{P}\left(\sum_{q=1}^{\bar M} (Z_q-\E\,Z_q) \geq t\bar M\right) \leq \exp \left(-\frac{2\bar M^2t^2}{\sum_{q=1}^{\bar M}(b_q-a_q)^2}\right)
    \end{align*}
\qed
\end{lemma}
\begin{theorem} \label{Theorem:NumberOfSamples}
   For any $\delta \in (0,1)$, if 
    \begin{align} \label{Equation:MlowerBound}
        M \geq \frac{1}{2(\epsilon-\alpha)^2}\log \left( \frac{1}{\delta}\right),
    \end{align}
then, a feasible control sequence $\{u_k\}_{k=0}^{N-1}$ w.r.t. Algorithm~\ref{algorithm:CIDA}, i.e. satisfies the inequalities \eqref{eq:statisticalFeasibility}, is feasible with respect to the original optimization problem \eqref{eq:optimizationProb}, with probability $1-\delta$.
\end{theorem}
\begin{proof}
Given an initial density $p_0$, fix a control sequence $\{u_k\}_{k=0}^{N-1}$, where $u_k \in \mathbb U$, and let $\{x_k\}_{k=0}^{N}$ be the resulting process characterized by \eqref{eq:NonSys_a}. Let $(\{x_{k,q}\}_{k=0}^{N})_{q=1}^M$ be defined as in Algorithm~\ref{algorithm:CIDA}, using this $\{u_k\}_{k=0}^{N-1}$.

Suppose $\{u_k\}_{k=0}^{N-1}$ is infeasible with respect to \eqref{eq:optimizationProb}. Then $\exists m\in \{1,\hdots,N\}$, a minimizer of $G_k = \mathbb P(x_k \in \mathcal{C})$ and hence $G_m = \mathbb P(x_m \in \mathcal{C})<1-\epsilon$. Define $Y_q = {1}(x_{l,q} \in \mathcal{C})$. Then
\begin{align*}
    \E Y_q = \mathbb P(x_{l,q} \in \mathcal{C}) = \mathbb P(x_l \in \mathcal{C})<1-\epsilon
\end{align*}
Since $x_{l,q}$ is an iid sample of $x_{l}$. 

Now we check the statistical feasibility of this control sequence, according to Algorithm~\ref{algorithm:CIDA},
\begin{align*}
        &\mathbb{P}\left(\left\{\frac{1}{M}\sum_{q=1}^M{1}(x_{k,q}\in\mathcal{C}) \geq 1-\alpha,\,\forall k\right \}\right),\\
        &\leq \mathbb{P}\left(\left\{\frac{1}{M}\sum_{q=1}^M{1}(x_{l,q}\in\mathcal{C}) \geq 1-\alpha\right \}\right),\\
        &\leq \mathbb{P}\left(\frac{1}{M}\sum_{q=1}^M \left ( Y_q - \E Y_q \right)\geq -1+\epsilon+1-\alpha\right),\\
        &\leq \mathbb{P}\left(\sum_{j=1}^M \left ( Y_q - \E Y_q \right)\geq M(\epsilon-\alpha)\right),\\
        &\leq \exp \left(-2 M(\epsilon-\alpha)^2\right),
\end{align*}
where the second inequality follows from the definition of $Y_q$, and the last inequality from Hoeffding's (Lemma~\ref{Hoeffdings}). If we pick $\delta \in (0,1)$, and pick an $M$ that satisfies \eqref{Equation:MlowerBound} with strict inequality, we get
\begin{align*}
   &\mathbb{P}\left(\left\{\frac{1}{M}\sum_{q=1}^M{1}(x_{k,q}\in\mathcal{C}) \geq 1-\alpha,\,\forall k\right \}\right),\\
   &\leq \exp \left(-2 M(\epsilon-\alpha)^2\right) < \delta.
\end{align*}
Therefore, if $\{u_k\}_{k=0}^{N-1}$ is infeasible w.r.t. \eqref{eq:optimizationProb}, then it is infeasible w.r.t. Algorithm~\ref{algorithm:CIDA}, with probability $\geq 1-\delta$. The contrapositive of this statement proves this theorem.
\end{proof}

Algorithm~\ref{algorithm:CIDA} has a computational complexity of $\mathcal{O}(MR)$, per time-step. By choosing $R$ to be of $\mathcal{O}(L)$, where $L$ represents the number of particles in the particle filter, the overall complexity of the algorithm, combined with the particle filter, becomes $\mathcal{O}(L \log L)$ per time-step. This algorithm is parallelizable across two dimensions: $i$, the random search trial number; and $q$, the index of the $M$ resulting trajectories. Using a graphics processing unit (GPU) can potentially reduce the required computation time.

\begin{remark}
    If $x_0' = \E x_0$ and $w_k'=0$ are used in \eqref{eq:x_k'def}, the resulting control sequence, call it $\{\bar u_k\}_k$, is the certainty equivalence control \cite{aastrom2012introduction}. Simply augmenting this sequence as one of the rollouts in Algorithm~\ref{algorithm:CIDA} guarantees that this algorithm is as good or better, in performance and safety over the prediction horizon, than the typically used certainty equivalence control.
\end{remark}
\begin{remark}
    The distribution of $w_k'$ determines the ``breadth'' of the searching space around the certainty equivalence control. Although $w_k'$ is presented as an i.i.d. sample of $w_k$, it can be sampled according to a different distribution. When there is no modeling mismatch: i.e. $w_k=0$ and $x_0$ is known, we expect the certainty equivalence control (with $w_k'=0$) to be optimal w.r.t. \eqref{eq:optimizationProb}. The more the uncertainties and modeling mismatches, the more the optimal solution will possibly depart from the certainty equivalence one, the more breadth we need for the distribution of $w_k'$, such that the optimal solution is still in the support of $\mathcal{U}$.
\end{remark}

\section{Numerical simulation: autonomous vehicle navigation and obstacle avoidance}\label{Section:examples}
An autonomous vehicle is modelled as a discrete-time stochastic unicycle model, with partial state observation:
\begin{align}
    \xi_{k+1}&=f(\xi_k,\omega_k,q_k), \nonumber \\
    &= \xi_k + \tau
    \begin{bmatrix}
    V \sinc (\frac{\omega_k \tau}{2})\cos(\theta_k+\frac{\omega_k \tau}{2})\\
    V \sinc (\frac{\omega_k \tau}{2})\sin(\theta_k+\frac{\omega_k \tau}{2})\\
    \omega_k
    \end{bmatrix}+w_k,\label{eq:stateDynamics}\\
    z_k&=g(\xi_k,v_k)= 
    \begin{bmatrix}
    x_k\\
    y_k
    \end{bmatrix}+ v_k\label{eq:measurementDynamics},
\end{align}
where $\tau=0.2\,s$ is the time interval of one step, $\sinc(\cdot)=\sin(\cdot)/\cdot$ is the sinc function, $\xi_k=(x_k,y_k,\theta_k)^T$ is the state vector, $\theta_k$ is the heading angle of the vehicle, counter clockwise from the positive $x-$axis, and $\omega_k \in \mathbb U = [-\pi,\pi]\,s^{-1}$ is the control input, which is the rate of change of heading angle $\theta_k$. We assume a constant vehicle speed $V = 5\,m/s$. The measurement can be seen as a noisy position measurement. The random variables $\zeta_0,w_k$ and $v_k$ follow the same assumptions for \eqref{eq:NonSys}, and: $\mathcal{W}=\mathcal{N}(0,\text{diag}(.2,.2,.1)),\,\mathcal{V}=\mathcal{N}(0,\text{diag}(.1,.1,.1))$ and the particles $\Xi$ representing $p_0$, the density of $\zeta_0$, are sampled according to a density $\mathcal{N}((10,0,-\pi/2)^T,\text{diag}(.2,.2,.2))$ \footnote{The notation: $\mathcal{N}(\mu,\Sigma)$ represents a multivariate Gaussian with mean $\mu$ and covariance $\Sigma$, $\text{diag}(e)$ is the square diagonal matrix with diagonal elements $e$.}.

This system has the objective to follow a circular orbit while avoiding multiple obstacles with probability $1-\epsilon$, where $\epsilon=15\%$. A CBF controller is well-suited to this objective for a deterministic system. However, the stochastic nature of the dynamics and the partial access to the states hinder the immediate implementation of a CBF-based control. Instead, we pick a surrogate deterministic dynamics model
\begin{equation}
\begin{aligned} \label{eq:uniCycleDeterministic}
\dot{p}_x = V_x,\, \dot{p}_y = V_y,
\end{aligned}
\end{equation}
where $p_x=p_x(t)$ and $p_y=p_y(t)$ are the coordinates of the vehicle, $V_x$ and $V_y$ are the velocity inputs.

\subsection{Level 1: baseline control $u_0$ for asymptotic behavior}
Suppose we use a baseline safety-agnostic controller, $u_0$, based on a vector field navigation approach \cite{nelson2007vector}. The vectors of this field represent the desired heading angle to follow an orbit. In this example, the path is a clockwise circular orbit with radius $r = 10\,m$ and a center at the origin. The baseline controller is then derived according to this vector field. The desired heading angle is defined as
\begin{align}
    \theta^d = \gamma - \frac{\pi}{2} - \tan^{-1}\left (k(d-r)\right), \label{eq:vecField}
\end{align}
where  $d$ and $\gamma$ are, respectively, the distance and the angular position of the vehicle with respect to the orbit's center, $\gamma = \atantwo(p_y,p_x)$. The angular position angle is measured from the positive $x-$axis, and in a counter-clockwise direction. We pick a gain $k=0.3$. The resulting vector field is visualized in Figure~\ref{fig:vectorField}.
\begin{figure}[h]
\centering 
\includegraphics[width=2.2in,height=2.2in]{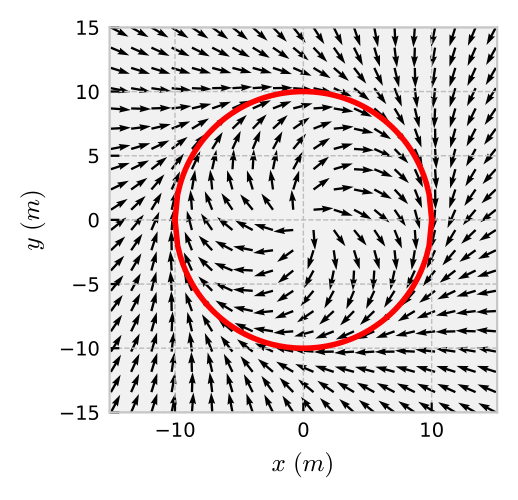} \caption{The vector field defined by a constant magnitude (speed) and an angle given by \eqref{eq:vecField}. The red curve is a circular orbit of $10\,m$ radius. \label{fig:vectorField}}
\end{figure}

The baseline control can be defined as $u_0 = (V \cos(\theta^d), V \sin(\theta^d))$, where $V=5\,m/s$ and $\theta^d$ as in \eqref{eq:vecField}.

\subsection{Level 2: CBF control $u^\star$ for obstacle avoidance}
While following the orbit, the vehicle has to avoid three obstacles, represented as circular objects with centers $\{(9,-5),(-10,-9),(-7,10)\}\,m$ and radii $\{3,4,3\}\,m$, respectively. For the multiple obstacles, we rely on the concept of the Boolean compositional nonsmooth barrier function \cite{glotfelter2017nonsmooth}.  The quadratic program, analogous to the one in Proposition~3 therein,
\begin{equation}
\begin{aligned} \label{CBF_QP_example}
u^\star(p_x,p_y) &= \argmin_u \lVert u - u_0\rVert^2,\\
& \hskip -10mm\text{s.t.:}\frac{\partial h_m (p_x,p_y)}{\partial (p_x,p_y)} ^T u \geq -\alpha(h_m(p_x,p_y)),\, m=1,2,3.
\end{aligned}
\end{equation}
The functions $h_m(x,y)$ correspond to each obstacle, and defined as
\begin{equation*}
    h_m(p_x,p_y) = (p_x-p_x^m)^2 + (p_y-p_y^m)^2 - (r^m)^2,
\end{equation*}
where $(p_x^m,p_y^m),r^m$ are the coordinates and radius of obstacle $m$. We pick the class $\mathcal{K}$ function $\alpha(h_m) = 0.05 h_m$ to induce a gradual variation in the vector field around these obstacles. The quadratic program is then solved for a grid over $[-15,15] \times [-15,15]$ in the state-space. The resulting $u^\star=(u^\star(1),u^\star(2))^T$ is mapped back to a desired heading angle, we denote it $\theta^\star=\atantwo(u^\star(2),u^\star(1))$. The resulting vector field is shown in Figure~\ref{fig:vectorField_obstacles}. This angle, since acquired using \eqref{CBF_QP_example}, is a function of the vehicle's position, i.e. $\theta^\star = \theta^\star(p_x,p_y)$.
\begin{figure}[h]
\centering 
\includegraphics[width=2.2in,height=2.2in]{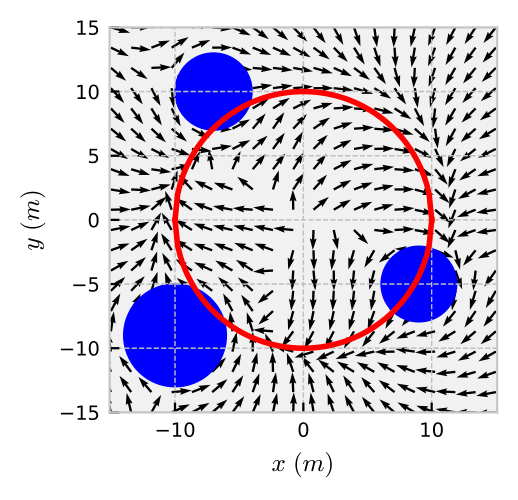} \caption{The vector field defined by $\theta^\star$, induced from the quadratic program \eqref{CBF_QP_example}. The red curve is a circular orbit of $10\,m$ radius, and the blue disks are the obstacles to be avoided by the vehicle. \label{fig:vectorField_obstacles}}
\end{figure}

\subsection{Level 3: CIDA control $u^\#$ for the stochastic system}
For the state-space system \eqref{eq:stateDynamics},\eqref{eq:measurementDynamics}, let $\hat \zeta_k=(\hat x_k, \hat y_k, \hat \theta_k)^T$ denote the sample average of the particle filtered density at time $k$. The number of particles in the filter we use is $L=1000$. We define the \textit{certainty equivalence controller} (CE), using the error signal between $\theta^\star(\hat x_k, \hat y_k)$ and $\hat \theta_k$, by $\omega_k^{\text{CE}} = \text{sat}\left ( 5(\theta^\star-\hat \theta_k),-\pi,\pi\right)$. This function saturates the first argument above at $\pi$ and below at $-\pi$, such that $\omega_k^{\text{CE}} \in \mathbb U=[-\pi,\pi],\,\forall k$. The result of applying this controller is depicted in Figure~\ref{fig:CBF_CE}, which shows the resulting true simulated trajectory that violates the constraints $54$ times over $750$ time-steps.
\begin{figure}[h]
\centering 
\includegraphics[width=2.2in,height=2.2in]{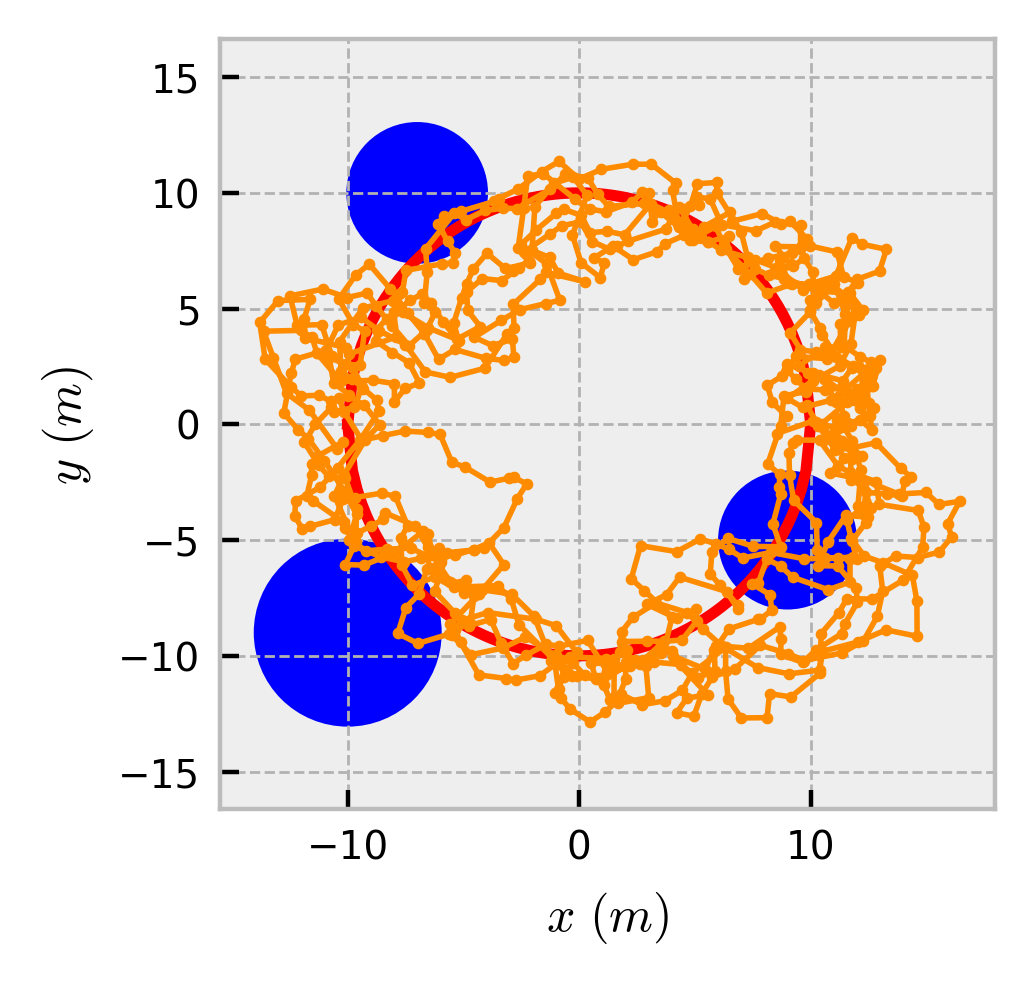} \caption{The true simulated vehicle's position for $750$ time-steps, when the certainty equivalence control is applied in feedback.}\label{fig:CBF_CE}
\end{figure}

Next, we apply our algorithm CIDA with: $\omega^\#=\pi(\zeta_0')$ where $\pi(\zeta_k') = \text{sat}\left ( 5(\theta^\star(x_k',y_k')-\hat \theta_k'),-\pi,\pi\right)$, $M=R=150,\,N=10,\,\gamma=1$, and statistical violation tolerance $\alpha=5\%$. According to Theorem~\ref{Theorem:NumberOfSamples}, these values enforce a probabilistic violation tolerance $\leq \epsilon=15\%$, and with confidence of at least $1-\delta = 95\%$. To avoid running into infeasibility issues, hard constraints are replaced by soft ones and the control sequence $\{u_k^\star\}_{k}$ in CIDA with the least $\max_k a_k$ is chosen.

Figure~\ref{fig:CBF_CIDA} illustrates a simulated trajectory using $\omega_k^\#$ in feedback. Compared to the certainty equivalence control (Figure~\ref{fig:CBF_CE}), CIDA takes into account the original dynamics \eqref{eq:NonSys} and their stochastic nature, and hence, is more capable of enforcing safety. With CIDA, the true simulated state violated the constraints for $15$ times compared to $54$ in CE, resulting in a controller that is $3.5$ times safer than the standard method.\footnote{The associated code:  \href{https://github.com/msramada/Control-Importance-Distribution-Algorithm-CIDA-}{https://github.com/msramada/Control-Importance-Distribution-Algorithm-CIDA-}}
\begin{figure}[h]
\centering 
\includegraphics[width=2.2in,height=2.2in]{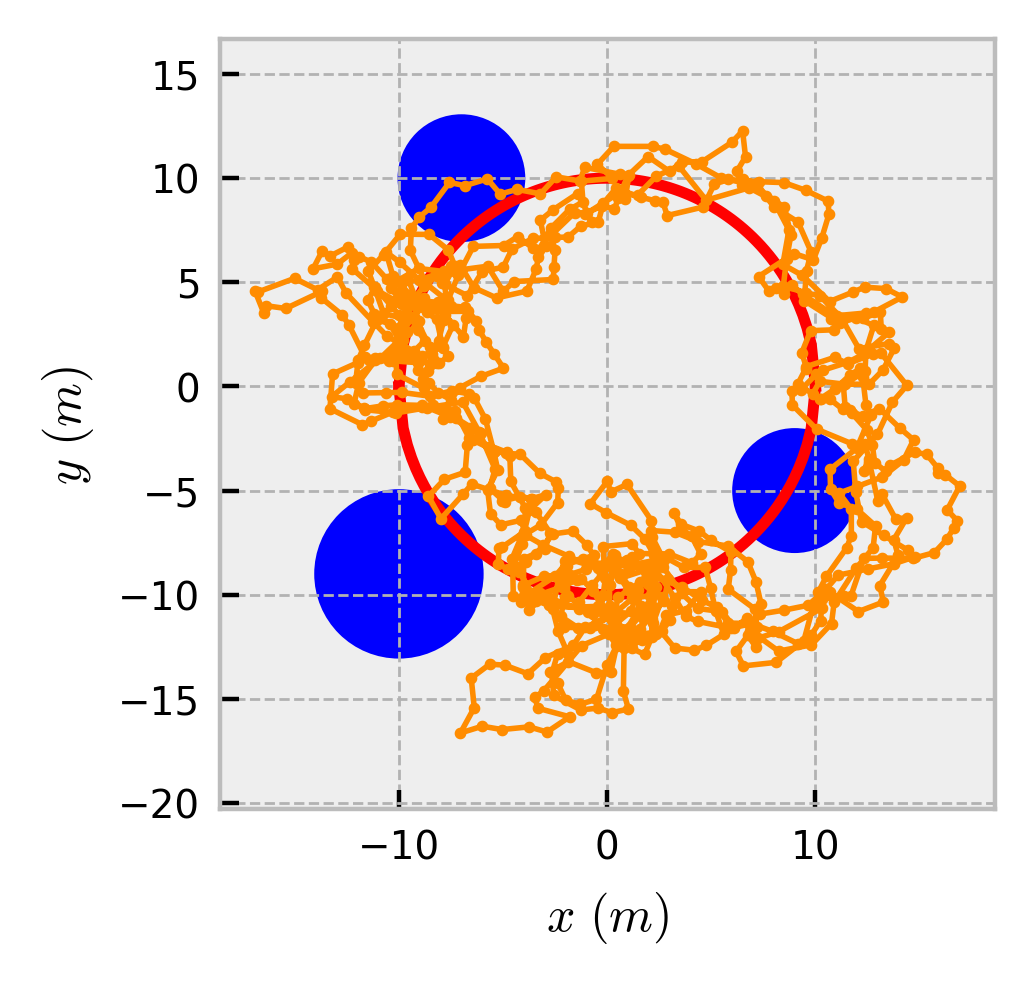} \caption{The true simulated vehicle's position for $750$ time-steps, when CIDA is implemented and the control $\omega^\#$ is applied in feedback.}\label{fig:CBF_CIDA}
\end{figure}

The simulations were implemented using \python, an interpretted programming language, on an M1-chip 2021 MacBook Pro with 16.00 GB of RAM. The average running time for each time-step of CIDA is $\approx 7$ seconds, compared to 0.07 seconds for CE. This can be vastly reduced via optimized parallel computing.
\section{Conclusion}
\label{section: Conclusion} 
This paper presents a method for generating probabilistically safe control for nonlinear stochastic state uncertain systems. Though generally intractable, we show that by limiting the search space based on a surrogate deterministic controller, a safe controller over a prediction horizon can be efficiently generated and guaranteed by a predefined margin $1-\epsilon$. We demonstrate the tradeoff between computational speed and safety compared with a standard certainty equivalence method. Algorithm~\ref{algorithm:CIDA} (CIDA) bears a passing resemblance to both model reference adaptive control \cite{parks1966liapunov}, and importance sampling from Monte Carlo integration literature \cite{doucet2000sequential}. This is due to the use of focused sampling of the control space, guided by a surrogate model that is both: ``close enough'' to the original stochastic dynamics and ``simple enough'' to admit nonlinear control design such as a CBF. However, these notions of closeness/simplicity have to be further investigated and, possibly, made quantifiable. Moreover, the recursive feasibility/safety condition is not discussed here; this is challenging in general for scenario/sampling-based methods, even with linearity and convexity assumptions adopted therein \cite{campi2019scenario}. More investigations of this condition is still required to extend the guarantees offered by these methods.

\bibliographystyle{IEEEtranS}
\bibliography{References}

\end{document}